\documentclass[11pt]{article}
\title{On Classical Integrability of the Hydrodynamics of Quantum Integrable Systems}

\usepackage{graphicx}
\usepackage{amsmath}
\usepackage{amsfonts}
\usepackage{amsthm}
\usepackage{amssymb}
\usepackage{amsbsy}
\usepackage{wasysym}
\usepackage{bm}
\usepackage{mathrsfs}

\usepackage{color}
\usepackage{hyperref}
\usepackage{fullpage}
\usepackage{bbm}
\usepackage{authblk}
\usepackage[square,numbers,sort&compress]{natbib}

\author[1]{Vir B. Bulchandani}
\affil[1]{\it{Department of Physics, University of California, Berkeley, Berkeley CA 94720, USA}}

\newcommand{\defeq}{\mathrel{\mathop:}=}
\newtheorem{defn}{Definition}
\newtheorem{conj}{Conjecture}
\newtheorem{theorem}{Theorem}
\newtheorem{prop}{Proposition}
\newtheorem{corollary}{Corollary}
\numberwithin{equation}{section}

\begin{document}
\maketitle
\nocite{apsrev41Control}
\bibliographystyle{apsrev4-1}
\begin{abstract}
Recently, a hydrodynamic description of local equilibrium dynamics in quantum integrable systems was discovered. In the diffusionless limit, this is equivalent to a certain ``Bethe-Boltzmann'' kinetic equation, which has the form of an integro-differential conservation law in $(1+1)$D. The purpose of the present work is to investigate the sense in which the Bethe-Boltzmann equation defines an ``integrable kinetic equation''. To this end, we study a class of $N$ dimensional systems of evolution equations that arise naturally as finite-dimensional approximations to the Bethe-Boltzmann equation. We obtain non-local Poisson brackets and Hamiltonian densities for these equations and derive an infinite family of first integrals, parameterized by $N$ functional degrees of freedom. We find that the conserved charges arising from quantum integrability map to Casimir invariants of the hydrodynamic bracket and their group velocities map to Hamiltonian flows. Some results from the finite-dimensional setting extend to the underlying integro-differential equation, providing evidence for its integrability in the hydrodynamic sense.
\end{abstract}

\tableofcontents
\section{Introduction}
\subsection{Background}
In the past year, a hydrodynamic description of dynamics in quantum integrable systems has emerged as a powerful tool for probing their out-of-equilibrium behaviour \cite{BB1,BB2,BB3,BB5,BB6,BB4,BB7,BB8}. Such systems are characterized by an infinity of local conservation laws, which give rise to unusual transport and thermalization properties. The insight underlying recent progress in the field is that making a local-density-type approximation for all of these local conserved charges implies a conservation law at the level of the local pseudo-momentum distribution \cite{BB1,BB2}. Moreover, the velocity associated with this pseudo-momentum turns out to be simply the single-excitation quasi-particle velocity, known from thermodynamic Bethe ansatz (TBA) \cite{YY,Takahashi}. The resulting kinematics takes the form of an integro-differential equation
\begin{equation}
\partial_t \rho(x,t,k) + \partial_x (\rho(x,t,k)v[\rho](k))=0,
\label{BBinit}
\end{equation}
where $v[\rho](k)$ denotes the quasi-particle velocity of the state with pseudo-momentum distribution $\{\rho(x,t,k) : k\in\mathbb{R}\}$. The intuitive meaning of this equation is that ``occupied quantum numbers are locally conserved''. Despite its simple formulation, this ``Bethe-Boltzmann'' equation has been found to capture local equilibrium dynamics in quantum integrable systems to a remarkable degree of accuracy. In the pioneering works \cite{BB1,BB2} on this equation, coordinates which diagonalize the system Eq. \eqref{BBinit} were discovered and used to obtain scale-invariant solutions. The natural Riemann invariants turn out to be the local Fermi factors, defined as the ratios $\theta(x,t,k) = \rho(x,t,k)/\rho^t(x,t,k)$ of the pseudo-momentum density to the total available density of states at every point. In terms of these, the Bethe-Boltzmann system assumes the diagonal form
\begin{equation}
\partial_t \theta(x,t,k) + v[\theta](k)\partial_x\theta(x,t,k)=0.
\label{seed}
\end{equation}
An efficient numerical scheme to solve this integro-differential equation (essentially a semi-Lagrangian method\footnote{We thank P.-O. Persson for this observation.}) was presented in an earlier work \cite{BB6}. In that paper, we also noted that a class of natural discretizations of the system \eqref{seed} exhibit surprising mathematical properties; for example, for every discretization dimension $N \geq 2$, the characteristic velocities are linearly degenerate and semi-Hamiltonian, and both of these observations can be extended to the underlying equation \eqref{seed}. The purpose of the present work is to frame these observations rigorously and develop their consequences in detail. In particular, we exploit the powerful theory of semi-Hamiltonian systems, developed by Tsar{\"e}v, Ferapontov and collaborators \cite{Tsarev85,Fera,FP,Bogd}, to obtain a family of non-local Poisson structures in the discrete setting, which allows us to construct a conjectural Poisson structure for the full, integro-differential dynamics \eqref{seed}.

The study of integrable kinetic equations has been an active area of research for some time, beginning with Gibbons's discovery in 1981 \cite{Gibbons} of the integrable Benney hydrodynamic chain \cite{Benney,KuperII,ZakharovBenney} within the moments of the Vlasov kinetic equation. This type of integrability was later characterized in terms of the existence of infinitely many semi-Hamiltonian reductions  \cite{GibbonsTsarev1,GibbonsTsarev2}, and the latter property was then formalized as part of a definition of integrability for $(2+1)$D quasilinear systems of PDEs \cite{FerapontovKhusnut}. This has since led to a fairly detailed understanding of the meaning of integrability for the Vlasov equation \cite{FerapontovMarshall,Vlasov,Odesskii}. More recently, a new type of kinetic equation was discovered in the physical context of the KdV soliton gas \cite{Zakharov,El1,El2}. This equation does not fit neatly into the existing notions of integrability for the Vlasov equation, but nonetheless exhibits a family of linearly degenerate, semi-Hamiltonian reductions \cite{El3}, which lie in the same equivalence class as the family of discretizations we obtain below.

The likely origin of this similarity is as follows. It turns out that the Bethe-Boltzmann equation Eq. \eqref{BBinit} can be written as the El-Kamchatnov equation \cite{El2} for a ``fictitious'' gas of classical solitons, provided we make the identification
\begin{equation}
\Delta x(k,k') = \frac{1}{p'(k)}\varphi'(k-k'),
\label{phaseident}
\end{equation}
where $\Delta x(k,k')$ denotes the position shift upon a collision of two classical solitons with spectral parameters $k$ and $k'$, $\varphi(k)$ denotes the two-particle phase shift of the underlying quantum integrable model, with rapidity parameter $k$, and $p(k)$ denotes the bare single-particle momentum. This mapping follows upon comparing the expressions given for the interaction-dressed velocity in Refs. \cite{BB1} and \cite{El2}, which both have the form
\begin{equation}
v(k) = v_0(k) + \int dk'\,\Delta x(k,k') [v(k)-v(k')]\rho(k'),
\label{EKdress}
\end{equation}
where $v_0(k)$ denotes the bare soliton or quasiparticle velocities. The kinetic theory interpretation of Eq. \eqref{EKdress}, noted by El and Kamchatnov \cite{El2}, was obtained independently in Ref. \cite{DS4}. The connection with classical solitons has also been elucidated in a corpus of recent work \cite{DS1,DS2,DS3} exploring analogies between the Bethe-Boltzmann equation and the dynamics of the classical hard-rod gas \cite{hardrod,SpohnBook}. In fact, the methods developed in the present work apply directly to the kinetic equation for the hard-rod gas, which is continuously linearly degenerate and semi-Hamiltonian in the sense of Prop. \ref{ctsWNS} below. It is less clear that such methods extend to the kinetic equation considered in Ref. \cite{El3}, because the logarithmic singularity in the KdV phase shift at $k=k'$ \cite{El4} leads to difficulties in the continuum limit.

This brings us to an important point. Whereas the El-Kamchatnov equation only appears to be integrable for a restricted class of classical soliton gases \cite{El3}, the structure of thermodynamic Bethe ansatz underlying the Bethe-Boltzmann equation implies that such integrability is \emph{generic} in the quantum case. Moreover, the Riemann invariants $\theta(x,t,k)$ are built into the definition of Eq. \eqref{BBinit} from TBA \cite{BB1,BB2}, whereas their construction for the El-Kamchatnov equation is non-trivial \cite{El3}. It is these pleasant properties, which are fundamental consequences of quantum integrability, that give rise to the integrability properties we obtain in the the continuum limit (see Section \ref{idl} below). This provides a partial answer to the question posed in Refs. \cite{El3,El4}: in what circumstances is the kinetic equation described by Eqs. \eqref{BBinit} and \eqref{EKdress} integrable?

\subsection{Guide to Results}

The paper is structured as follows. In Section \ref{sec2}, we summarize the required results on diagonal systems of hydrodynamic type. In Section \ref{sec3}, we exhibit a class of systems with the general structure of finite-dimensional approximations to Eq. \eqref{seed}. In Proposition \ref{mainprop}, we prove that these systems are linearly degenerate and semi-Hamiltonian and in Theorem \ref{mainthm}, we obtain a non-local Poisson bracket \cite{FeraNL} of the form
\begin{align}
\{I,J\} = \sum_{i,j} \int dx\,\frac{\delta I}{\delta \theta^i(x)} A^{ij} \frac{\delta J}{\delta \theta^j(x)}
\end{align}
where
\begin{align}
\label{DNF}
A^{ij} = g^{ii}\delta^{ij}\frac{d}{dx} - g^{ii}\Gamma^{j}_{ik}\theta^k_x + \sum_{m=1}^M\epsilon_m\chi^{(m)}_i\theta^i_x d_x^{-1}\chi^{(m)}_j\theta^j_x.
\end{align}
Explicit formulas for the various coefficients in this bracket are given in Section \ref{sec3}. In Section \ref{sec4}, we deduce some corollaries for the Euler-scale hydrodynamics of the Lieb-Liniger gas. In particular, we show that in each $N$-point discretization, the charges arising from quantum integrability $\{\mathbf{Q}^{(n)}\}_{n=0}^{\infty}$ give rise to sets of independent Casimirs \cite{maltsev} $\{\mathcal{Q}^{(n)}\}_{n=0}^{N-1}$ and Hamiltonians $\{H^{(n)}\}_{n=0}^{N-1}$, which satisfy the algebra
\begin{equation}
\{\mathcal{Q}^{(m)},\mathcal{Q}^{(n)}\}=0, \quad \{H^{(m)},\mathcal{Q}^{(n)}\}=0, \quad \{H^{(m)},H^{(n)}\}=0, \quad m,n = 0,1,\ldots,N-1.
\label{chargealg}
\end{equation}
We obtain these quantities explicitly, including the Hamiltonian $H^{(1)}$ that generates the discretized evolution equations $\partial_t \theta^i + v_i \partial_x \theta^i=0$. We find that several of the results proved in the discretized setting extend to the integro-differential limit, including suitably generalized notions of linear degeneracy, the semi-Hamiltonian property and the generalized hodograph transform. This leads us, in Conjecture \ref{conj1}, to propose a ``continuum'' Poisson bracket of the form
\begin{align}
\{I,J\} = \int dx \int dk \, dk' \frac{\delta I}{\delta \theta(x,k)} A(k,k') \frac{\delta J}{\delta \theta(x,k')},
\end{align}
where the operator-valued kernel $A(k,k')$ is given by
\begin{align}
\label{ctsbrackprev}
\nonumber A(k,k') = &\frac{\delta(k-k')}{P'(k)^2}\left[\frac{d}{dx} - \int dk''\, \frac{\alpha(k,k'')P'(k'')}{P'(k)}\theta_x(x,k'')\right] + \frac{\alpha(k,k')}{P'(k)P'(k')}(\theta_x(x,k)-\theta_x(x,k'))\\
- &\frac{\theta_x(x,k)}{P'(k)}\int dk''\,\alpha(k,k'')d_x^{-1}\alpha(k'',k')\frac{1}{P'(k')}\theta_x(x,k').
\end{align}
Here, $P(k)$ denotes the dressed quasiparticle momentum and $\alpha(k,k')$ is defined in Eq. \eqref{kern} below. This appears to define an infinite-dimensional version of the Ferapontov bracket \cite{FeraNL} in the same sense that the bracket introduced for the Vlasov equation in Ref. \cite{Vlasov} defines an infinite-dimensional version of the Dubrovin-Novikov bracket. In Proposition \ref{sensible}, the bracket \eqref{ctsbrackprev} is found to generalize some key properties of the finite-dimensional Hamiltonian structure, in particular the algebra \eqref{chargealg}. In contrast to the $N$-dimensional brackets, which possess at most $N$ linearly independent Hamiltonian flows, the continuum bracket conjectured in the integro-differential limit exhibits \emph{countably} many such flows. The formulation of these results for more general quantum integrable models is given in the Appendix. 

We now summarize the physical meaning of these results. The Casimirs $\mathcal{Q}^{(n)}$ alluded to above turn out to be precisely the conserved charges $\mathbf{Q}^{(n)}$ of the underlying quantum system, as computed in a local density approximation \cite{BB1,BB2}, while each Hamiltonian $H^{(n)}$ is the generator of a ``dressed group velocity'' for the charge $\mathbf{Q}^{(n+1)}$. In a different vein, we find that for any discretization dimension $N$, the Bethe-Boltzmann evolution is obtainable as the image of the non-interacting evolution under a generalized reciprocal transformation (Section \ref{recip}), of a type originating in classical gas dynamics \cite{FP,recip}. This is related to the existence of a new, uncountable family of first integrals (Proposition \ref{intprop}) of the discretized Bethe-Boltzmann equation, whose analogues we are able to construct for the full integro-differential system \eqref{seed}.

Before proceeding, we ought to comment on our various usages of the term  ``integrability''. Here, a ``quantum integrable model'' is a Bethe-ansatz solvable quantum system in one spatial dimension, with diagonal scattering. Meanwhile, the finite-dimensional approximations to the Bethe-Boltzmann equation are ``integrable'' in the sense that they can be obtained from a reciprocal transformation of free-particle trajectories. Finally, we conjecture that the integro-differential system \eqref{seed} is ``integrable'' in the same sense as its finite-dimensional approximants.

\section{Definition and Properties of Semi-Hamiltonian Systems}
\label{sec2}
We begin by summarizing the main results in the theory of semi-Hamiltonian systems. Thus consider the non-linear system of PDEs
\begin{equation}
\partial_t \theta^i + v_i(\theta^1,\theta^2,\ldots,\theta^N) \partial_x\theta^i = 0.
\label{sys}
\end{equation}
This defines a diagonal, quasilinear system of evolution equations in $1+1$ dimensions. The space $\mathbb{R}^N$ of vectors $(\theta^1,\theta^2,\ldots,\theta^N)$ will henceforth be referred to as the \emph{target space}, and the space of fields $\theta^i : \mathbb{R}^2 \to \mathbb{R}^N$ will be called the \emph{phase space}. The study of local Hamiltonian structures for such ``systems of hydrodynamic type'' was initiated by Dubrovin and Novikov \cite{Dub}. Soon afterwards, Tsar{\"e}v discovered a weakening of the local Hamiltonian property which retained some vestiges of integrability.
\begin{defn}
(Tsar{\"e}v \cite{Tsarev85})
The system of equations \eqref{sys} is called a \emph{semi-Hamiltonian system} if the characteristic vector field $v_i$ possesses the semi-Hamiltonian property
\begin{equation}
\partial_k\left(\frac{\partial_j v_i}{v_j-v_i}\right) = \partial_j \left(\frac{\partial_kv_i}{v_k-v_i}\right) \quad i\neq j \neq k \neq i.
\end{equation}
where $\partial_j v_i \defeq \partial v_i/\partial \theta^j$.
\end{defn}
For such systems, the natural notion of Poisson bracket, due to Dubrovin-Novikov \cite{Dub} and defined below, generically fails to satisfy the Jacobi identity. However, it was conjectured by Ferapontov \cite{Bogd,FeraNL} that all semi-Hamiltonian systems \textit{are} in fact Hamiltonian with respect to certain non-local extensions of the Dubrovin-Novikov bracket. Moreover, such systems exhibit various remarkable properties analogous to integrability. For example, they may be solved by quadrature and possess infinitely many first integrals of motion. These properties are related to the following class of symmetries.
\begin{theorem}
\label{flowthm}
(Tsar{\"ev} \cite{Tsarev85})
For a semi-Hamiltonian system, upon defining coefficients
\begin{equation}
\Gamma^i_{ij} = \frac{\partial_j v_i}{v_j-v_i}
\end{equation}
for $i\neq j$, the overdetermined system of equations
\begin{equation}
\partial_jw_i = \Gamma^i_{ij}(w_j-w_i), \quad i \neq j 
\label{flow}
\end{equation}
is consistent. This has infinitely many solutions $w_i(\theta)$ which define flows commuting with the system \eqref{sys} and with each other, parameterized by $N$ functions of a single variable. We shall call these the \emph{hydrodynamic symmetries} of the system under consideration.
\end{theorem}
Such flows give rise to the ``generalized Hodograph transform'' method for solving semi-Hamiltonian systems; the main result is as follows.
\begin{theorem}
\label{hod}
(Tsar{\"e}v \cite{Tsarev85})
Each flow $w_i(\theta)$ solving \eqref{flow} can be used to define a system of equations
\begin{equation}
w_i(\theta) = x-v_i(\theta)t.
\end{equation}
Solving these for $\theta$ yields a solution to \eqref{sys}, and locally, every smooth solution arises in this way.
\end{theorem}
Thus the local structure of smooth solutions to semi-Hamiltonian systems is fully characterized. We note that the coefficients $\Gamma^i_{ij}$ may be associated with a diagonal metric on the target space, via the formula
\begin{equation}
\Gamma^i_{ij} = \partial_j \log \sqrt{g_{ii}}.
\label{christ}
\end{equation}
The $\Gamma^i_{ij}$ are the Christoffel symbols associated with the metric $g_{ii}$ defined by \eqref{christ}. The unifying theme in the theory of systems of hydrodynamic type is that in coordinates $\theta^j$ such that the system \eqref{sys} is diagonal, the differential geometry of the metric $g_{ii}$ dictates the allowed Hamiltonian structures on the space of fields. This has lead to some beautiful analogies between the geometry of surfaces and the study of hydrodynamic reductions of integrable PDEs \cite{Dub,FeraNL,Mokh}. The next result characterizes first integrals of the semi-Hamiltonian evolution. We first define a class of relevant functionals.
\begin{defn}
A \emph{hydrodynamic functional} on the space of fields is an expression of the form $H = \int dx\, h(\theta^1,\theta^2,\ldots,\theta^N)$, where $h$ depends solely on $\theta$ and not its higher derivatives $\theta_x,\theta_{xx},\ldots$. 
\end{defn}
We then have the following result.
\begin{theorem}
\label{firstint}
(Tsar{\"e}v \cite{Tsarev85})
A hydrodynamic functional $H$ with density $h(\theta)$ defines a first integral of the evolution \eqref{sys} if and only if $h(\theta)$ satisfies the overdetermined system of equations
\begin{equation}
\partial_i\partial_jh - \Gamma^i_{ij}\partial_ih - \Gamma^j_{ji}\partial_jh= 0
\end{equation}
for $i \neq j$. This is consistent by the semi-Hamiltonian property, and gives rise to infinitely many linearly independent first integrals, parameterized by $N$ functional degrees of freedom.
\end{theorem}
It is useful to have a criterion for deciding when the system \eqref{sys} is truly Hamiltonian. The first such criterion to be discovered was the following geometrical statement.
\begin{theorem}
\label{LocalH}
(Dubrovin-Novikov \cite{Dub})
The system \eqref{sys} is Hamiltonian with respect to the Dubrovin-Novikov bracket 
\begin{equation}
\{I,J\} = \sum_{i,j} \int dx\,\frac{\delta I}{\delta \theta^i(x)} A^{ij} \frac{\delta J}{\delta \theta^j(x)}
\end{equation}
of functionals $(I,J)$ on phase space, where
\begin{equation}
A^{ij} = g^{ii}\delta^{ij}\frac{d}{dx} - g^{ii}\Gamma^{j}_{ik}\theta^k_x,
\end{equation}
if the elements
\begin{equation}
R^j_{iij} = \partial_i \Gamma^j_{ji} - \partial_j \Gamma^j_{ii} + \Gamma^j_{ji}\Gamma^j_{ji} + \Gamma^i_{ij}\Gamma^j_{ii} - \sum_k \Gamma^k_{ii}\Gamma^j_{jk}
\end{equation}
of the Riemann curvature tensor associated with $g_{ii}$ vanish.
\end{theorem}
This ensures that $g_{ii}$ is a zero-curvature metric on the target space, and consequently that the Dubrovin-Novikov bracket associated with the metric $g_{ii}$ satisfies the Jacobi identity, and so defines a bona-fide Poisson bracket on phase space. However, Mokhov and Ferapontov \cite{FeraNL,Mokh} have shown that for certain semi-Hamiltonian systems with non-vanishing curvature, there exist non-local Poisson brackets with respect to which these systems \emph{are} Hamiltonian. The criterion we shall need for the existence of such brackets is as follows.
\begin{theorem}
\label{NonLocalH}
(Ferapontov \cite{FeraNL})
The semi-Hamiltonian system \eqref{sys} is Hamiltonian with respect to the non-local Poisson bracket
\begin{equation}
A^{ij} = g^{ii}\delta^{ij}\frac{d}{dx} - g^{ii}\Gamma^{j}_{ik}\theta^k_x + \sum_{m=1}^M\epsilon_m\chi^{(m)}_i\theta^i_x d_x^{-1}\chi^{(m)}_j\theta^j_x
\end{equation}
if there exists a family of hydrodynamic symmetries $\{\chi^{(m)}_i\}_{m=1}^M$, possibly infinite, such that the diagonal elements of the Riemann tensor $R^{ij}_{ij} = -g^{jj}R^j_{iij}$ can be expressed as 
\begin{equation}
R^{ij}_{ij} =  \sum_{m=1}^M \epsilon_m\chi^{(m)}_i\chi^{(m)}_j.
\end{equation}
Here, $d_x^{-1}$ denotes the inverse differential operator and the $\epsilon_m = \pm1$.
\end{theorem}
This amounts to a rather subtle completeness requirement for squares $w_iw_j$ of solutions to \eqref{flow}. It has been shown that formally speaking, any semi-Hamiltonian system possesses this property, even though infinitely many flows $\chi_m$ might be required \cite{Bogd}. Maltsev and Novikov \cite{maltsev} have developed the theory of the resulting ``weakly non-local'' brackets in detail. Finally, we will be concerned with the following subset of ``solvable'' semi-Hamiltonian systems, for which a non-local Poisson bracket can be constructed explicitly \cite{FP}.
\begin{defn}
The system of equations \eqref{sys} is a \emph{weakly non-linear semi-Hamiltonian system} if it is semi-Hamiltonian, and the characteristic velocities are additionally linearly degenerate, in the sense that $\partial_j v_j=0$ for all $j$.
\end{defn}
For such ``WNS systems'', the $N$ functional degrees of freedom parameterizing flows and first integrals of the system \eqref{sys} can be obtained directly using the methods of Ferapontov. This is possible thanks to the existence of a special class of hydrodynamic symmetries.
\begin{prop}
\label{WNS}
(Ferapontov \cite{Fera})
A WNS system possesses exactly $N$ linearly independent, commuting WNS vector fields solving \eqref{flow}, of which two are the trivial flows $w^{(0)}_i = 1$ and $w_i^{(1)} = v_i$.
\end{prop}
These special flows yield solutions to the system \eqref{flow} as follows.
\begin{prop}
\label{quad}
(Ferapontov \cite{Fera})
The full set of hydrodynamic symmetries of any WNS system is given by
\begin{equation}
w_i(\theta) = \Delta^0(\theta) - v_i(\theta)\Delta^{1}(\theta) - \sum_{n=2}^{N-1}w^{(n)}_i(\theta)\Delta^n(\theta),
\end{equation}
where 
\begin{equation}
\Delta^k(\theta) = \int^{\theta^1} \frac{d\xi\, \phi_1^k(\xi)}{f_1(\xi)} + \int^{\theta^2} \frac{d\xi\, \phi_2^k(\xi)}{f_2(\xi)} + \ldots + \int^{\theta^N} \frac{d\xi\, \phi_N^k(\xi)}{f_N(\xi)}.
\end{equation}
Here, the functions $\phi^n_k(\xi)$ are fixed uniquely by the vector field $v_i$, whilst the functions $f_i(\xi)$ are arbitrary, and encode the $N$ functional degrees of freedom alluded to in Theorem \ref{flowthm}.
\end{prop}
This concludes our brief survey of this elegant theory. We make no claim to completeness, and our presentation is tailored towards the results to be derived below.

\section{A Class of Hamiltonian Systems}
\label{sec3}
\subsection{Definition; Proof of WNS Property}
Consider the system of PDEs
\begin{equation}
\partial_t\theta^i + v_i[\theta]\theta^i = 0
\label{curr}
\end{equation}
with characteristic velocities $v_i$ defined as follows. Choose linearly independent vectors $\mathbf{a},\mathbf{b} \in \mathbb{R}^N$, a real, symmetric $N$-by-$N$ matrix $K_{ij}$, let $T_{ij}(\theta) = K_{ij}\theta^j$ and define the dressing operator 
\begin{equation}
U_{ij}(\theta) = (1 + T(\theta))^{-1}_{ij}.
\end{equation}
Let $\mathcal{S}$ be a neighbourhood of $0 \in \mathbb{R}^N$ such that the series
\begin{equation}
U_{ij} = \delta_{ij} - T_{ij}(\theta) + \sum_k T_{ik}(\theta)T_{kj}(\theta) + \ldots
\end{equation}
converges. On $\mathcal{S}$, we define ``dressed'' vectors $A_i(\theta) = \sum_jU_{ij}(\theta)a_j$ and $B_i(\theta) = \sum_jU_{ij}(\theta)b_j$. Choosing $\mathbf{a}$ such that $A_i > 0$ on $\mathcal{S}$, we define the characteristic velocities in \eqref{curr} to be
\begin{align}
v_i(\theta) =\frac{B_i(\theta)}{A_i(\theta)}.
\label{vels}
\end{align}
This is the generic form of the equations arising in a class of finite approximations to the Bethe-Boltzmann equation for any quantum integrable model with diagonal scattering, introduced in \cite{BB1,BB2}. They lie in the equivalence class of ``Egorov'' \cite{Egor} linearly-degenerate semi-Hamiltonian systems previously obtained for the KdV soliton gas in Ref. \cite{El3}. Before specializing to a particular model, we prefer to derive some key properties of the system given by Eqs. \eqref{curr} and \eqref{vels} in generality. We first have the following result.
\begin{prop}
Suppose that $A_i >0$ on $\mathcal{S}$ for $i = 1,2,...,N$. Then the system \eqref{curr} with velocities \eqref{vels} defines a WNS system on $\mathcal{S}$.
\label{mainprop}
\end{prop} 
\begin{proof} 
First note that 
\begin{align}
\nonumber \partial_jU_{ik}(\theta) = -\sum_{l,m} U_{il}(\theta)\partial_j(K_{lm}\theta_m)U_{mk}(\theta) = -\sum_l U_{il}K_{lj}U_{jk},
\end{align}
where we henceforth suppress arguments. Defining $\alpha_{ij} = -\sum_l U_{il}K_{lj}$, we have
\begin{align}
\nonumber \partial_jU_{ik} = \alpha_{ij}U_{jk},
\end{align}
which implies
\begin{align}
\partial_jA_i = \alpha_{ij}A_j,
\end{align}
and similarly for $B$. Thus
\begin{align}
\nonumber \partial_jv_i &= \frac{(\partial_j B_i)A_i -B_i(\partial_jA_i)}{(A_i)^2} = \frac{\alpha_{ij}(B_jA_i - B_iA_j)}{(A_i)^2}.
\end{align}
We deduce that $\partial_iv_i = 0$ for all $i$. Moreover, for $i \neq j$, 
\begin{align}
\nonumber v_j -v_i = \frac{B_jA_i -B_iA_j}{A_iA_j},
\end{align}
so that
\begin{align}
\frac{\partial_jv_i} {v_j -v_i} = \frac{\alpha_{ij}A_j}{A_i} = \frac{\partial_jA_i}{A_i} = \partial_j \log{A_i}.
\label{value}
\end{align}
Similarly, $$\frac{\partial_kv_i}{v_k -v_i} = \partial_k \log(A_i).$$ Thus
\begin{align}
\nonumber \partial_k \left(\frac{\partial_jv_i}{v_j -v_i}\right) = \partial_j \left(\frac{\partial_kv_i}{v_k -v_i}\right) = \partial_j\partial_k \log(A_i),
\end{align}
which was to be shown.
\end{proof}
By a judicious choice of $\mathbf{a}$, such a system exists for any symmetric matrix $K$, in some neighbourhood $\mathcal{S}$ of the point $\theta = 0$. In view of the theory of semi-Hamiltonian systems, we have also determined various other important properties of the system \eqref{curr}. Firstly, we can read off a diagonal Riemannian metric
\begin{equation}
g_{ii} = A_i^2
\end{equation}
on the target space, with respect to which 
\begin{equation}
\Gamma^i_{ij} = \partial_j \log(A_i)
\end{equation}
define Christoffel symbols. Thus the system of equations \eqref{flow} takes the form
\begin{equation}
\frac{\partial_jw_i}{w_j-w_i} = \partial_j \log(A_i), \quad j \neq i
\label{flow2}
\end{equation}
in the present context. Since the derivation of the result \eqref{value} was independent of the choice of constant vector $\mathbf{b}$, we deduce the following result.
\begin{prop}
Take $N$ linearly independent vectors $\mathbf{a}^{(n)} \in \mathbb{R}^N$ such that $\mathbf{a}^{(0)} = \mathbf{a}$ and $\mathbf{a}^{(1)} = \mathbf{b}$. Define
\begin{equation}
A^{(n)}_i(\theta) = \sum_j U_{ij}(\theta)a^{(n)}_j, \quad n=0,1,\ldots,N-1.
\end{equation}
Then the vector fields $w^{(n)}$, given by
\begin{align}
w^{(n)}_i (\theta) = \frac{A^{(n)}_i(\theta)}{A_i(\theta)}, \quad n=0,1,\ldots,N-1,
\label{wnsflows}
\end{align}
define $N$ linearly independent, commuting WNS flows on $\mathcal{S}$, in the sense of Prop. \ref{WNS}. 
\label{localwns}
\end{prop}
It will greatly aid intuition for the remainder of this section to regard $a^{(n)}$ as the derivative of the $n$th bare charge density of some integrable model with respect to rapidity. As one might expect, their flows are related to certain first integrals of the system \eqref{curr}.
\begin{prop}
\label{wnscharges}
Consider the functionals
\begin{equation}
Q^{(n)}[\theta] = \int dx \, \sum_i \theta^i A_i(\theta) a^{(n)}_i, \quad n=0,1,\ldots,N-1
\end{equation}
of hydrodynamic type. These define $N$ first integrals of the system \eqref{curr} with functionally independent densities.
\end{prop}
\begin{proof}
Write
\begin{equation}
\rho^{(n)} = \sum_i \theta^i A_i(\theta) a^{(n)}_i.
\end{equation}
Then
\begin{align}
\partial_j \rho^{(n)} 
= A_j\sum_i (\alpha_{ji}\theta^i+\delta_{ji})a^{(n)}_i 
=A_j\sum_i U_{ji}a^{(n)}_i 
= A_jA^{(n)}_j,
\end{align}
where we used the definition of $\alpha$ and the crucial fact (deducible from the series expansion of $U$) that $\alpha_{ij}=\alpha_{ji}$. Thus for $i \neq j$,
\begin{align}
\nonumber &\partial_i\partial_j\rho^{(n)} -\Gamma^i_{ij}\partial_i\rho^{(n)}-\Gamma^j_{ji}\partial_j\rho^{(n)} \\
\nonumber =\,&\alpha_{ij}\left[A_jA^{(n)}_i+A_iA^{(n)}_j - A_jA^{(n)}_i - A_iA^{(n)}_j\right] \\
\nonumber =\,&0.
\end{align}
It follows by Theorem \ref{firstint} that these densities give rise to first integrals of \eqref{curr}. Moreover, they are related to the above WNS flows via $w^{(n)}_i = g^{ii}\partial_i \rho^{(n)}$, which implies linear independence of their gradients.
\end{proof}
We note that the density $\rho^{(0)}$ defines a potential for the metric, since $g_{ii} = A_i^2 = \partial_i \rho^{(0)}$ \cite{Egor}. Thus the system under consideration is of Egorov type, and lies in the isomorphism class derived in Ref. \cite{El3}.


\subsection{Existence of non-local Poisson Brackets}
We now compute the curvature associated with the metric $g_{ii}$ under consideration. Specifically, we compute the non-vanishing components of curvature $R^j_{iij}$ of the Riemann curvature tensor associated with $g$, namely
\begin{align}
R^j_{iij} = \partial_i \Gamma^j_{ji} - \partial_j \Gamma^j_{ii} + \Gamma^j_{ji}\Gamma^j_{ji} + \Gamma^i_{ij}\Gamma^j_{ii} - \sum_k \Gamma^k_{ii}\Gamma^j_{jk},
\end{align}
since it is these which determine whether the hydrodynamic brackets of Theorems \ref{LocalH} and \ref{NonLocalH} define Poisson brackets. One can determine these directly from the Christoffel symbols, but it is quickest to use existing results  \cite{FP,FeraNL} and define \emph{Lam{\'e} coefficients} and \emph{rotation coefficients} by
\begin{align}
H_i = \sqrt{g_{ii}}, \quad \beta_{ij} = \frac{\partial_i H_j}{H_i}, \quad i \neq j
\end{align}
respectively. Then the diagonal components of the Riemann tensor can be expressed as
\begin{equation}
R^j_{iij} = \frac{H_j}{H_i}\left(\partial_i\beta_{ij} + \partial_j \beta_{ji} + \sum_{k \neq i,j}\beta_{ki}\beta_{kj}\right).
\end{equation}
In the present case, we have
\begin{equation}
H_i = A_i, \quad \beta_{ij} = \alpha_{ij}
\end{equation}
and it is straightforward to determine that
\begin{equation}
R^{j}_{iij} = \frac{A_i}{A_j}\sum_{k=1}^N \alpha_{ik}\alpha_{jk}
\label{curv}
\end{equation}
for $i \neq j$. This does not generically vanish, and by Theorem \ref{LocalH}, the Dubrovin-Novikov bracket is not Hamiltonian for our system. However, by Theorem \ref{NonLocalH}, there remains the possibility that our system is Hamiltonian with respect to its non-local extension, provided we can find a family of commuting flows $\{\chi^{(m)}\}_{m=1}^M$ and $\epsilon_m = \pm 1$, such that the diagonal elements of the Riemann tensor $R^{ij}_{ij} = -g^{jj}R^j_{iij}$ satisfy
\begin{equation}
R^{ij}_{ij} = -\frac{1}{A_iA_j}\sum_{k=1}^N \alpha_{ik}\alpha_{jk} = \sum_{m=1}^M \epsilon_m \chi^{(m)}_i\chi^{(m)}_j
\end{equation}
for $i\neq j$. Finding suitable $\chi^{(m)}$ turns out to be remarkably simple.
\begin{theorem}
\label{mainthm}
Let $\{\chi^{(m)}\}_{m=1}^N$ denote $N$ WNS flows generated by constant vectors $c^{(m)}_j = -K_{jm}$, as in Eqn. \eqref{wnsflows}. Then
\begin{equation}
R^{ij}_{ij} = \sum_{m=1}^N \epsilon_m\chi^{(m)}_i\chi^{(m)}_j
\end{equation}
with $\epsilon_m=-1$, and the system \eqref{curr} is Hamiltonian with respect to the non-local Poisson bracket
\begin{equation}
A^{ij} = g^{ii}\delta^{ij}\frac{d}{dx} - g^{ii}\Gamma^{j}_{ik}\theta^k_x - \sum_{m=1}^N\chi^{(m)}_i\theta^i_x d_x^{-1}\chi^{(m)}_j\theta^j_x.
\label{bracket}
\end{equation}
\end{theorem}
\begin{proof}
Let $c^{(m)}_j = -K_{jm}$. Define
\begin{equation}
C^{(m)}_i = \sum_j U_{ij}c^{(m)}_j = -\sum_j U_{ij}K_{jm} = \alpha_{im}
\end{equation}
and
\begin{equation}
\chi^{(m)}_i= \frac{C_i^{(m)}}{A_i} = \frac{\alpha_{im}}{A_i}.
\end{equation}
Each of these vector fields has the structure of the flows considered in Prop. \ref{localwns}, and therefore defines a weakly non-linear semi-Hamiltonian flow (although they are no longer required to be linearly independent).  In particular,
\begin{align}
\nonumber \sum_{m=1}^N \chi^{(m)}_i\chi^{(m)}_j &= \sum_{m=1}^N \frac{\alpha_{im}\alpha_{jm}}{A_iA_j} = -R^{ij}_{ij}.
\end{align}
The result follows by Theorem \ref{NonLocalH}.
\end{proof}
There is a geometrical interpretation \cite{FeraNL} of this property of $R^{ij}_{ij}$; namely, that there exists a realization of the metric $g = \sum_i g_{ii}(\theta)d\theta^i\otimes d\theta^i$ as a net of lines of curvature of some $N$-dimensional surface with flat normal connection in the pseudo-Riemannian space $\mathbb{R}^{N+N}$, all of whose normals are ``spatially similar''.


\subsection{Infinite Families of Symmetries and First Integrals}
\label{commflows}
Following Ferapontov \cite{Fera}, we now use the WNS flows to obtain the infinite family of hydrodynamic symmetries of the system \eqref{curr}. Thus consider an $N$-dimensional submanifold of the target space traced out by evolution along the commuting WNS flows $w^{(0)},w^{(1)},\ldots,w^{(N-1)}$ given in Eq. \eqref{wnsflows}. We can parametrize these flows by coordinates $t_1,\ t_2,\ldots,\ t_{N-1}$, giving rise to the system of evolution equations
\begin{align}
\nonumber \partial_{t_1}\theta^i + v_i(\theta)\partial_x\theta^i &= 0 \\
\nonumber \partial_{t_2}\theta^i + w^{(1)}_i(\theta)\partial_x\theta^i &= 0 \\
\nonumber &\vdots\\
\partial_{t_{N-1}}\theta^i + w^{(N-1)}_i(\theta)\partial_x\theta^i &= 0.
\label{commsys}
\end{align}
Rather than tackling this system directly, we let $f_i(\theta_i)$ denote $N$ arbitrary functions of a single variable and consider the system
\begin{align}
\nonumber \partial_x \theta^i &= f_i(\theta^i)A^{(0)}_i(\theta) \\
\partial_{t_{n}}\theta^i &= - f_i(\theta^i)A^{(n)}_i(\theta), \quad n=1,2,\ldots,N-1,
\label{commsys2}
\end{align}
which implies $\eqref{commsys}$. We can summarize this as
\begin{align}
\nonumber \frac{\partial_x \theta^i}{f_i(\theta^i)} &= A^{(0)}_i(\theta) \\
\frac{\partial_{t_n}\theta^i}{f_i(\theta^i)} &= -A^{(n)}_i(\theta), \quad n=1,2,\ldots,N-1,
\end{align}
whose differential version reads
\begin{equation}
\frac{d\theta^i}{f_i(\theta^i)} = A^{(0)}_i(\theta)dx - \sum_{n=1}^{N-1} A^{(n)}_i(\theta)dt_n.
\end{equation}
In order to integrate this expression, one must separate variables; this is achieved by acting with the inverse dressing operator to yield
\begin{equation}
\sum_j \frac{(\delta_{ij} + K_{ij}\theta^j)d\theta^j}{f_j(\theta^j)} = a^{(0)}_idx - \sum_{n=1}^{N-1} a^{(n)}_idt_n.
\end{equation}
Integrating, we obtain the implicit equation
\begin{equation}
\sum_j \int^{\theta^j} \frac{(\delta_{ij} + K_{ij}\xi)d\xi}{f_j(\xi)} = a^{(0)}_i x - \sum_{n=1}^{N-1} a^{(n)}_i t_n.
\label{crucialint}
\end{equation}
Finally, we act with the operator $U_{ij}/A_i$ to obtain the ``generalized hodograph'' form
\begin{equation}
\frac{1}{A_i(\theta)}\sum_{j,k}(1+K\theta)^{-1}_{ij}\int^{\theta^k} \frac{(\delta_{jk} + K_{jk}\xi)d\xi}{f_k(\xi)} = x - \sum_{n=1}^{N-1} w^{(n)}_i t_n
\end{equation}
for solutions of the system \eqref{commsys}. One can deduce that the flows solving \eqref{flow2} are given in terms of $N$ functional degrees of freedom $f_k(\xi)$ by
\begin{equation}
w_i(\theta) = \frac{1}{A_i(\theta)}\sum_{j,k}U_{ij}(\theta)\int^{\theta^k} \frac{(\delta_{jk} + K_{jk}\xi)d\xi}{f_k(\xi)}.
\label{result}
\end{equation}
There are ways to derive this formula which make its relation with the result of Proposition \ref{quad} more explicit \cite{Fera}. However, these are more involved and all yield the same final result \eqref{result}. For the sake of completeness, we note that the quadratures referred to in Proposition \ref{quad} are given by
\begin{align}
\Delta^{n-1} &= \sum_k \int^{\theta^k} \sum_{j} \frac{{q'}^{-1}_{nj}(\delta_{jk} + K_{jk}\xi)d\xi}{f_k(\xi)}, \quad n=1,2,\ldots,N
\end{align}
where the matrix $q'$ has elements
\begin{align}
\label{chargematrix}
q' = \begin{pmatrix}
a_1^{(0)} & -a_1^{(1)} & \ldots & -a_1^{(N-1)} \\
a_2^{(0)} & -a_2^{(1)} & \ldots & -a_2^{(N-1)} \\
\vdots & \vdots & \ddots & \vdots \\
a_N^{(0)} & -a_N^{(1)} & \ldots & -a_N^{(N-1)}
\end{pmatrix}.
\end{align}
We now formalize the result \eqref{result}, which clarifies its relation with the WNS flows obtained above.
\begin{prop}
\label{hydrosym}
The system \eqref{curr} has infinitely many hydrodynamic symmetries,
\begin{equation}
w_i(\theta) = \frac{\sum_j U_{ij}(\theta)c_j(\theta)}{A_i(\theta)}
\end{equation}
solving the system \eqref{flow}, and parameterized by $N$ functional degrees of freedom $\{f_k(\xi)\}_{k=1}^N$ via
\begin{equation}
c_j(\theta) = \sum_k\int^{\theta^k} \frac{(\delta_{jk} + K_{jk}\xi)d\xi}{f_k(\xi)}.
\label{gencharge}
\end{equation}
\end{prop}
\begin{proof}
Note that
\begin{align}
\partial_k c_j = \frac{(\delta_{jk} + K_{jk}\theta^k)}{f_k(\theta^k)}.
\end{align}
Thus, defining $C_i = \sum_j U_{ij}c_j$, we have
\begin{align}
\partial_k C_i = \sum_j \alpha_{ik}U_{kj}c_j + U_{ij}(\delta_{jk} + K_{jk}\theta^k)\frac{1}{f_k(\theta^k)} = \alpha_{ik}C_k + \frac{\delta_{ik}}{f_i(\theta_i)}.
\end{align}
In particular, for $i\neq k$, $\partial_k C_i = \alpha_{ik}C_k$, and
\begin{equation}
\frac{\partial_k w_i}{w_k-w_i} = \frac{\alpha_{ik}A_k}{A_i} = \Gamma^i_{ik}
\end{equation}
by the derivation given in Prop. \ref{mainprop}.
\end{proof}
Associated with these symmetries are the first integrals referred to in Prop. \ref{firstint}.
\begin{prop}
\label{intprop}
The system \eqref{curr} has infinitely many first integrals of hydrodynamic type, with densities
\begin{equation}
h(\theta) = \sum_i \left[\theta^iA_i(\theta)c_i(\theta) - a^{(0)}_i\int^{\theta^i} \frac{d\xi \, \xi}{f_i(\xi)}\right]
\end{equation}
solving the system \eqref{flow}, and parameterized by $N$ functional degrees of freedom $\{f_k(\xi)\}$ via Eq. \eqref{gencharge}. Their gradients $g^{ii}\partial_i h$ recover the hydrodynamic symmetries obtained above.
\end{prop}
\begin{proof}
One can show that
\begin{align}
\nonumber \partial_j h = A_jC_j - a_j^{(0)}\theta^j\frac{1}{f_j(\theta^j)} + \sum_i \theta^iA_i(\delta_{ij} + K_{ij}\theta^j)\frac{1}{f_j(\theta^j)}.
\end{align}
Now
\begin{equation}
\sum_i \theta^iA_i(\delta_{ij} + K_{ij}\theta^j) = \sum_i A_i(\delta_{ij}\theta^i + \theta^iK_{ij}\theta^j) = \theta^j\sum_i (\delta_{ji} + K_{ji}\theta^i)A_i = a^{(0)}_j\theta^j
\end{equation}
by symmetry of $K$. Thus
\begin{equation}
\partial_j h = A_jC_j.
\end{equation}
Since $\partial_i C_j = \alpha_{ji}C_i$ for $i\neq j$, it follows by the proof of Prop. \ref{wnscharges} that $h$ is the density of a first integral. It is also clear that $g^{ii}\partial_ih = C_i/A_i = w_i$.
\end{proof}


\subsection{The Hamiltonian}
Now that we have fully characterized the hydrodynamic symmetries and associated first integrals of our system, let us turn to the question of finding the Hamiltonian giving rise to the dynamics \eqref{curr} with respect to our non-local Poisson bracket \eqref{bracket}. First, it is helpful to write down the equations of motion
\begin{equation}
\nonumber \partial_t \theta^i = \{\theta^i,H\}
\end{equation}
associated with the Ferapontov bracket explicitly, where $H$ is a functional of hydrodynamic type with density $h$. In the present context, we find that
\begin{align}
\nonumber \{\theta^i,H\} = &\frac{1}{A_i^2}\sum_{j\neq i} \left[\frac{\partial^2 h}{\partial \theta^i \partial \theta^j} -\nonumber \frac{\alpha_{ij}A_j}{A_i}\frac{\partial h}{\partial \theta^i}-\frac{\alpha_{ji}A_i}{A_j}\frac{\partial h}{\partial \theta^j} \right]\theta^j_x +\frac{1}{A_i^2}\left[\frac{\partial^2 h}{\partial \theta^i \partial \theta^i} - \alpha_{ii}\frac{\partial h}{\partial \theta^i} + \sum_{j \neq i} \frac{\alpha_{ij}A_i}{A_j}\frac{\partial h}{\partial \theta^j}\right]\theta^i_x \\
& - \frac{1}{A_i}\sum_m \alpha_{im}d_x^{-1}\left(\alpha_{mj}\theta^j_x\frac{1}{A_j}\frac{\partial h}{\partial \theta^j}\right)\theta^i_x.
\end{align}
In particular, whenever $H$ is a first integral of motion, so that $\nabla^i\nabla_j h = 0$ for $i\neq j$, we obtain the diagonal expression
\begin{equation}
\{\theta^i,H\} = \frac{1}{A_i}\left(\frac{1}{A_i}\left[\frac{\partial^2 h}{\partial \theta^i \partial \theta^i} - \alpha_{ii}\frac{\partial h}{\partial \theta^i} + \sum_{j \neq i} \frac{\alpha_{ij}A_i}{A_j}\frac{\partial h}{\partial \theta^j}\right] - \sum_m \alpha_{im}d_x^{-1}\left(\alpha_{mj}\theta^j_x\frac{1}{A_j}\frac{\partial h}{\partial \theta^j}\right)\right)\theta^i_x.
\label{hambrack}
\end{equation}
In order to remove any ambiguity in the operator $d_x^{-1}$, it is convenient to work on the compactified real line. A discussion of more general boundary conditions and their associated subtleties \cite{maltsev,Sklyanin} will be deferred to later work.
\begin{prop}
\label{theham}
\begin{enumerate} 
\item The first integrals
\begin{equation}
\mathcal{Q}^{(n)}[\theta] = \int dx \, \sum_i \theta^i A_i(\theta) a^{(n)}_i, \quad n=0,1,\ldots,N-1,
\label{Cassie}
\end{equation}
arising from the $WNS$ flows, lie in the kernel of the bracket \eqref{bracket}.
\item The Hamiltonian
\begin{equation}
H^{(1)}[\theta] = -\int dx \, \sum_i \left(\theta^iA_i(\theta)c^{(1)}_i(\theta) - \frac{1}{2}a^{(0)}_ia^{(1)}_i(\theta^i)^2 \right)
\end{equation}
with
\begin{equation}
c^{(1)}_i(\theta) = \sum_j \left(\delta_{ij}\theta^j + \frac{1}{2}K_{ij}(\theta^j)^2\right)a^{(1)}_j
\end{equation}
gives rise to the evolution \eqref{curr}.
\end{enumerate}
\end{prop}
\begin{proof}
\begin{enumerate}
\item Let $\rho^{(n)}(\theta) = \sum_i \theta^i A_i(\theta) a^{(n)}_i$. It was shown above that $\partial_j \rho^{(n)} = A_jA_j^{(n)}$. Substituting into \eqref{hambrack} yields
\begin{align}
\nonumber \{\theta^i, \mathcal{Q}^{(n)}\} &= \frac{1}{A_i}\left(\frac{1}{A_i}\left[2\alpha_{ii}A_iA_i^{(n)} - \alpha_{ii}A_iA_i^{(n)} + \sum_{j \neq i} \alpha_{ij}A_iA_j^{(n)}\right] - \sum_m \alpha_{im}d_x^{-1}\left[\alpha_{mj}\theta^j_xA_j^{(n)}\right]\right)\theta^i_x \\
\nonumber &= \frac{1}{A_i}\left(\sum_{j} \alpha_{ij}A_j^{(n)} - \sum_m \alpha_{im}d_x^{-1}d_x[A^{(n)}_m]\right)\theta^i_x\\
\nonumber &= \frac{1}{A_i}\left(\sum_{j} \alpha_{ij}A_j^{(n)} - \sum_m \alpha_{im}A^{(n)}_m\right)\theta^i_x \\
&=0.
\end{align}
\item Now consider the density
\begin{equation}
h^{(1)}(\theta) = -\sum_i \left(\theta^iA_i(\theta)c^{(1)}_i(\theta) - \frac{1}{2}a^{(0)}_ia^{(1)}_i(\theta^i)^2 \right).
\end{equation}
By taking $1/f_i(\xi)= a^{(1)}_i$ in Prop. \ref{intprop}, we find that $\partial_i h^{(1)} = -A_iC_i$, with $C_i = \sum_j U_{ij}c_j$. Then by various arguments given above,
\begin{align}
\nonumber \{\theta^i, H^{(1)}\} &= -\frac{1}{A_i}\left(\frac{1}{A_i}\left[2\alpha_{ii}A_iC_i + a^{(1)}_i - \alpha_{ii}A_iC_i + \sum_{j \neq i} \alpha_{ij}A_iC_j\right] - \sum_m \alpha_{im}d_x^{-1}\left[\alpha_{mj}\theta^j_xC_j\right]\right)\theta^i_x \\
\nonumber &= -\frac{1}{A_i}\left(\sum_{j} \alpha_{ij}C_j + a^{(1)}_i - \sum_m \alpha_{im}d_x^{-1}d_x \left[C_m-a^{(1)}_m\theta^m\right]\right)\theta^i_x\\
\nonumber &= -\frac{1}{A_i}\left(\sum_{j} \alpha_{ij}C_j + a^{(1)}_i - \sum_m \alpha_{im}C_m +\sum_m \alpha_{im}a^{(1)}_m\theta^m\right)\theta^i_x \\
\nonumber &= -\frac{1}{A_i}\left(\sum_m (\delta_{im}+\alpha_{im}\theta^m)a^{(1)}_m\right) \\
&= -\frac{B_i}{A_i}\theta^i_x.
\end{align}
Thus the system \eqref{curr} can be expressed in terms of the non-local Poisson bracket \eqref{bracket} and the above Hamiltonian as
\begin{equation}
\partial_t \theta^i = \{\theta^i, H^{(1)}\}.
\end{equation}
\end{enumerate}
\end{proof}
We remark that the generator of the $n$th WNS flow is the Hamiltonian
\begin{equation}
H^{(n)}[\theta] = - \int dx \sum_i \, \left(\theta^iA_i(\theta)c^{(n)}_i(\theta) - \frac{1}{2}a^{(0)}_ia^{(n)}_i(\theta^i)^2 \right),
\end{equation}
with
\begin{equation}
c^{(n)}_i(\theta) = \sum_j \left(\delta_{ij}\theta^j + \frac{1}{2}K_{ij}(\theta^j)^2\right)a^{(n)}_j,
\end{equation}
and this expression can be simplified to
\begin{equation}
H^{(n)}[\theta] = - \int dx \sum_i \, \frac{1}{2}(\theta^i)^2A_i(\theta)a^{(n)}_i,
\end{equation}
which is closer to the form of the Casimirs, Eq. \eqref{Cassie}, albeit less useful for proving the results above. It is not hard to check that in the absence of spatial boundaries, the algebra of Casimirs and Hamiltonians is given by
\begin{equation}
\{\mathcal{Q}^{(m)},\mathcal{Q}^{(n)}\}=0, \quad \{H^{(m)},\mathcal{Q}^{(n)}\}=0, \quad \{H^{(m)},H^{(n)}\}=0, \quad m,n = 0,1,\ldots
\label{falg}
\end{equation}
By Propositions \ref{localwns} and \ref{wnscharges}, the Casimirs $\{\mathcal{Q}^{(n)}\}_{n=0}^{N-1}$ possess functionally independent densities, while the Hamiltonians $\{H^{(n)}\}_{n=0}^{N-1}$ generate linearly independent flows at each point. Thus, in an $N$ point discretization, only the first $N$ of the relations \eqref{falg} are linearly independent.


\subsection{Linearization via a Reciprocal Transformation}
\label{recip}
We close our discussion of the system \eqref{curr} with some remarks on how it arises from an underlying linear problem. It has been shown \cite{FP} that the non-local Poisson structure for any WNS system can be constructed as the image of a local Hamiltonian structure under a suitably defined reciprocal transformation. In particular, evolution along WNS flows may always be regarded as the image of an underlying \emph{linear} evolution under a suitable reciprocal transformation. We now exhibit such a transformation for the present system, following Ferapontov and Pavlov \cite{FP}. Thus consider the system of linear equations
\begin{equation}
\partial_t \theta^i + \frac{a^{(1)}_i}{a^{(0)}_i} \partial_x \theta^i = 0.
\end{equation}
This is trivially Hamiltonian, with a local Poisson bracket $A^{ij} = g^{ij}\frac{d}{dx}$ of hydrodynamic type, flat metric $g^{ij} =\frac{\delta^{ij}}{(a^{(0)}_i)^2}$ and Hamiltonian density
\begin{equation}
h(\theta) = -\frac{1}{2} \sum_{i} a^{(0)}_i a^{(1)}_i (\theta^i)^2.
\end{equation}
Now introduce variables $\{s_1,s_2,\ldots,s_N\}$ and consider the system of linear flows
\begin{equation}
\partial_{s_n}\theta^i = v^{(n)}_i\partial_{s_1}\theta^i, \quad n=1,2\ldots,N,
\label{linsys}
\end{equation}
with
\begin{align}
\nonumber v^{(1)}_i &= 1 \\
v^{(n)}_i &= -\frac{a^{(n-1)}_i}{a^{(0)}_i}, \quad n=2,\ldots,N.
\end{align}
We consider the ``generalized reciprocal transformation''
\begin{equation}
d\tilde{s}_i = r_{in}ds_n
\label{infrec}
\end{equation}
with
\begin{equation}
r_{in} = [{q'}^{-1}(1+K\theta)q']_{in}.
\label{recipmat}
\end{equation}
One can verify directly that
\begin{align}
\partial_{s_m} r_{i1} = \partial_{s_1} r_{im}, \quad m=1,2,\ldots,N.
\end{align}
Thus the $r_{i1}$ are conserved densities of the system \eqref{linsys} and the $r_{im}$ are their fluxes. The inverse transformation of $r$, satisfying
\begin{equation}
ds_n = R_{nj}d\tilde{s}_j
\end{equation}
has components
\begin{equation}
R_{nj} = [{q'}^{-1}Uq']_{nj}
\end{equation}
and can be used to evaluate the transformed velocity fields, via
\begin{align}
\nonumber v^{(n)}_i = \frac{R_{mn}v_i^{(m)}}{R_{m1}v_i^{(m)}} = \frac{(q'R)_{in}}{(q'R)_{i1}} &= \frac{(Uq')_{in}}{(Uq')_{i1}} = \begin{cases} 
1 & n=1\\
-\frac{A^{(n-1)}_i}{A^{(0)}_i}  & n=2,\ldots,N
\end{cases}.
\end{align}
These are precisely the WNS flows derived above. This confirms the validity of the generalized reciprocal transformation defined by Eqs. \eqref{infrec} and \eqref{recipmat}. One can additionally check that the curvature and non-local bracket obtained in this fashion coincide with results \eqref{curv} and \eqref{bracket} given above. This result has a rather striking physical interpretation in the context of quantum integrable models.


\section{Applications to Bethe-Boltzmann Hydrodynamics}
\label{sec4}
\subsection{Bethe-Boltzmann Theory of the Lieb-Liniger Gas}
Let us now recall some basic definitions of the Bethe-Boltzmann hydrodynamics of quantum integrable models. For a discussion of the underlying physics, the reader is referred to the existing literature \cite{BB1,BB2}. For concreteness, we work in the context of the Lieb-Liniger gas.\footnote{The main difference in applying this theory to more general integrable models lies in the number of quasiparticle excitations; see Appendix. The Lieb-Liniger gas also has the marked advantage that its conserved charges are easily expressed in terms of pseudo-momenta.} The structure of this equation is as follows: we have a function $\theta(x,t,k)$ depending on position, time and pseudo-momentum $k\in \mathbb{R}$, and satisfying the integro-differential equation
\begin{equation}
\partial_t\theta(x,t,k) + v[\theta](k)\partial_x\theta(x,t,k) = 0
\label{BB}
\end{equation}
where
\begin{equation}
v[\theta](k) = \frac{\int_{-\infty}^{\infty} dk' \, U(k,k')k'}{\int_{-\infty}^{\infty} dk' \, U(k,k')}.
\end{equation}
Here, the dressing operator $U$ is determined in terms of the local Fermi factors $\{\theta(x,t,k): k \in \mathbb{R}\}$ by the integral equation
\begin{equation}
U(k,k') + \int_{-\infty}^{\infty} dk''\,K(k,k'')\theta(k'')U(k'',k') = \delta(k-k').
\label{dress}
\end{equation}
Some important physical quantities are the pseudo-momentum density $\rho(x,t,k)$, the total density of states $\rho^t(x,t,k)$ and the local Fermi factor $\theta(x,t,k)$, which are related by
\begin{align}
\rho^t(x,t,k) = \frac{1}{2\pi}\int_{-\infty}^{\infty} dk' \, U(k,k'), \quad \rho(x,t,k) = \rho^t(x,t,k)\theta(x,t,k).
\end{align}
The conserved charges of the underlying quantum model have densities
\begin{equation}
q^{(n)}(k) = \frac{k^n}{n}, \quad n \in \mathbb{N},
\end{equation}
and give rise to spatial densities of charges and currents
\begin{align}
\nonumber \rho^{(n)}(x,t) &= \int_{-\infty}^{\infty} dk \, \rho(x,t,k)q^{(n)}(k) \\
j^{(n)}(x,t) &= \int_{-\infty}^{\infty} dk \, \rho(x,t,k)q^{(n)}(k)v(x,t,k).
\label{LDA}
\end{align}
One can show that the Bethe-Boltzmann equation is implied by local density approximations \eqref{LDA}, in combination with the infinite family of conservation laws
\begin{equation}
\partial_t \rho^{(n)}(x,t) + \partial_x j^{(n)}(x,t) = 0, \quad n \in \mathbb{N}.
\end{equation}
It is natural to expect that these conservation laws resurface in the semi-Hamiltonian geometry of the underlying kinetic equation; we shall see that this is indeed the case. Finally, we note that every ``bare'' charge gives rise to a state-dependent ``dressed'' charge, and the $k$-derivatives ${q^{(n)}}'$ and ${Q^{(n)}}'$ of bare and dressed charge densities are related by the integral equation \cite{Takahashi}
\begin{equation}
{Q^{(n)}}'(k) + \int_{-\infty}^{\infty} dk' \,\mathcal{K}(k,k')\theta(k'){Q^{(n)}}'(k') = {q^{(n)}}'(k).
\end{equation}
In particular, the quasiparticle velocity appearing in \eqref{BB} is defined to equal $v(k)={Q^{(2)}}'(k)/{Q^{(1)}}'(k)$.

\subsection{A Class of Hamiltonian Discretizations}
\label{comp}
For each $N \geq 2$, we can construct an $N$-dimensional discretization of \eqref{BB} as follows. Introduce a k-space cut-off $\Lambda$ and choose $N+1$ evenly spaced $k$-space points $-\Lambda/2=k_0 < k_1 <  \ldots < k_N = \Lambda/2$. Then functions map to vectors and kernels map to matrices; for example, we have discrete analogues
\begin{align}
\nonumber q_i &= q(k_i)\\
\nonumber \theta^i(x,t) &= \theta(x,t,k_i)\\
K_{ij} &= \frac{\Lambda}{N}\mathcal{K}(k_i,k_j)
\end{align}
of charges, Fermi factors and the Lieb-Liniger kernel respectively, where $i,j =1,2,\ldots,N$. Let us also introduce the matrix $T_{ij}(\theta) = K_{ij}(\theta)\theta^j$. Then upon discretizing, derivatives of dressed charge densities $Q'$ are related to bare ones $q'$ via the matrix equation
\begin{equation}
\sum_j (\delta_{ij} + T_{ij}(\theta))Q'_j(\theta) = q'_i.
\end{equation}
Let $\mathcal{S}$ denote a neighbourhood\footnote{It is worth noting that for large $N$, e.g. $N = \mathcal{O}(1000)$, there is numerical evidence that $\mathcal{S}$ is large enough for all practical purposes; for example, in the numerical evolutions considered in \cite{BB6}, the dressing operator was always found to possess a convergent Neumann series.} of $\theta=0$ in $\mathbb{R}^N$ on which the dressing operator $U_{ij}(\theta) = (1 + T(\theta))^{-1}_{ij}$ exists. On $\mathcal{S}$, we define dressed derivatives of energy and momenta $E'_i(\theta) =\sum_j U_{ij}(\theta)k_j,\ P'_i(\theta) =\sum_j U_{ij}(\theta)$, and the discretized quasiparticle velocities can be written as 
\begin{equation}
v_i(\theta) = E'_i(\theta)/P'_i(\theta)
\label{DBB1}
\end{equation}
In terms of these functions, the discretized Bethe-Boltzmann equation reads
\begin{equation}
\partial_t\theta^i + v_i(\theta)\partial_x\theta^i = 0, \quad i = 1,2,...,N
\label{DBB2}
\end{equation}
and has precisely the WNS structure considered in Section \ref{sec3}. Thus it can be characterized in terms of a target space metric
\begin{equation}
g_{ii} = (P'_i)^2,
\end{equation}
with respect to which the coefficients
\begin{equation}
\Gamma^i_{ij} = \frac{\partial_j v_i}{v_j-v_i}
\end{equation}
define Christoffel symbols. It also possesses $N$ WNS flows, a convenient basis for which is obtained by dressing the vectors
\begin{equation}
a^{(n-1)}_i = {q^{(n)}}'_i = k_i^{n-1}, \quad n =1,2,\ldots,N
\end{equation}
to yield the vector fields
\begin{equation}
w^{(n-1)}_i = \frac{{Q^{(n)}}'_i}{P'_i}, \quad n=1,2,\ldots,N.
\end{equation}
It is again useful to introduce the matrix
\begin{equation}
\alpha_{ij} = -\sum_kU_{ik}K_{kj},
\end{equation}
and the constant matrix $q'$ as in Eq. \eqref{chargematrix}. We now formulate the most important corollaries of Section \ref{sec3}.
\begin{corollary}
\label{maincorr}
\begin{enumerate}
\item For any finite discretization length $N \geq 2$, the discretized Bethe-Boltzmann system, \eqref{DBB1} and \eqref{DBB2}, is Hamiltonian on $\mathcal{S}$, with respect to the non-local Poisson bracket
\begin{equation}
A^{ij} = g^{ii}\delta^{ij}\frac{d}{dx} - g^{ii}\Gamma^{j}_{ik}\theta^k_x - \sum_{m=1}^N\chi^{(m)}_i\theta^i_x d_x^{-1}\chi^{(m)}_j\theta^j_x,
\label{discbrack}
\end{equation}
where the affinors $\{\chi^{(m)}\}_{m=1}^N$ are obtained by dressing the columns of the Lieb-Liniger kernel, namely
\begin{equation}
\chi^{(m)}_j(\theta) = \frac{\sum_k U_{jk}(\theta)K_{km}}{P'_j(\theta)}.
\end{equation}
\item With respect to this bracket, the system \eqref{DBB2} can be written in Hamiltonian form as
\begin{equation}
\partial_t \theta^i = \{\theta^i,H^{(1)}\},
\end{equation}
where
\begin{equation}
H^{(1)}[\theta] = -\pi \int dx \sum_i \rho_i(x) \theta^i(x) k_i.
\end{equation}
\item The system \eqref{DBB2} has infinitely many hydrodynamic symmetries, given in terms of $N$ functional degrees of freedom $f_k(\xi)$ by
\begin{equation}
w_i(\theta) = \frac{1}{P'_i(\theta)}\sum_{j,k}U(\theta)_{ij}\int^{\theta^k} \frac{(\delta_{jk} + K_{jk}\xi)d\xi}{f_k(\xi)}.
\label{genflow}
\end{equation}
Moreover, every smooth solution to this system is given locally by a solution to the generalized hodograph equation
\begin{equation}
w_i(\theta) = x-v_i(\theta)t
\end{equation}
for some $w_i(\theta)$ of the form \eqref{genflow}.
\item The system \eqref{DBB2} is the image under a generalized reciprocal transformation of the free particle evolution
\begin{equation}
\partial_t \theta^i + k_i \partial_x \theta^i = 0,
\end{equation}
where the transformation in question, mapping the system of $N$ non-interacting flows
\begin{equation}
\partial_{s_n} \theta^i + {q^{(n)}}'_i \partial_{s_1} \theta^i = 0, \quad n=2,\ldots,N
\end{equation}
to the $N$ WNS flows
\begin{equation}
\partial_{\tilde{s}_n} \theta^i + \frac{{Q^{(n)}}'_i}{P'_i} \partial_{\tilde{s}_1} \theta^i = 0, \quad n=2,\ldots,N
\end{equation}
is given by
\begin{equation}
d\tilde{s}_i = r_{in}ds_n
\end{equation}
with
\begin{equation}
r_{in} = [{q'}^{-1}(1+K\theta)q']_{in}.
\end{equation}
\end{enumerate}
\end{corollary}
\begin{proof}
\begin{enumerate}
\item Follows by Theorem \ref{mainthm}.
\item Follows by Prop. \ref{theham}.
\item Follows by Prop. \ref{hydrosym}, in combination with the Theorem \ref{hod} of Tsar{\"e}v.
\item Follows by the discussion in Section \ref{recip}.
\end{enumerate}
\end{proof}
We note that from a physical viewpoint, the last of these results is the most interesting. In particular, the limit in which the reciprocal transformation $r$ becomes the identity is precisely the free Fermion limit $K=0$. Thus, the discretized Bethe-Boltzmann evolution can be obtained by a geometrical transformation of its free-Fermion counterpart. Moreover, the non-linearity of this transformation is tuned by the model interaction strength. This yields a surprising relationship between the classical geometry of the discretized Bethe-Boltzmann equations and the physics of quasiparticle dressing.

We also note that the above result holds for any $N$ whatsoever. In particular, the discretized Bethe-Boltzmann equation is Hamiltonian and arises as the image of a free-Fermion evolution under a reciprocal transformation for any finite $N$. Meanwhile, upon taking $N,\Lambda \to \infty$ in our choice of discretization, the matrix equations satisfied by $E'_i$ and $P'_i$ tend to the integral equations arising in TBA \cite{YY}, under rather weak assumptions on the continuum Fermi factor $\theta^i$. We believe that this points towards integrability of the underlying integro-differential equation in the sense of the above results.

Before moving on, it is worth specifying the relationship between the conserved charges of the quantum integrable model, the Poisson structure described above, and the dimension $N$ of the discretization. From Proposition \ref{theham} and the subsequent remarks, we deduce the following.

\begin{corollary}
\label{corr2}
\begin{enumerate}
\item The $N$ first integrals
\begin{equation}
\mathcal{Q}^{(0)}[\theta] = \int dx \sum_i \frac{\Lambda}{N} \rho_i(x), \quad \mathcal{Q}^{(n)}[\theta] = \int dx \sum_i \frac{\Lambda}{N} \rho_i(x) q_i^{(n)}, \quad n=1,2,\ldots,N-1,
\end{equation}
define Casimir invariants of the bracket \eqref{discbrack} with functionally independent densities. These integrals correspond to the total particle number and the first $N-1$ higher conserved charges respectively.
\item The $N$ first integrals $\{H^{(0)}, H^{(1)},\ldots,H^{(N-1)}\}$, with
\begin{equation}
H^{(n)}[\theta] = -\pi \int dx\, \sum_i \rho_i(x)\theta^i(x){q^{(n)}}'_i,
\end{equation}
generate the $N$ commuting flows $w^{(0)},w^{(1)},\ldots,w^{(N-1)}$ at each point. These flows correspond to group velocities for the dressed charges.
\item The first $N$ Casimirs and Hamiltonians, as defined above, satisfy the algebra
\begin{equation}
\{\mathcal{Q}^{(m)},\mathcal{Q}^{(n)}\}=0, \quad \{H^{(m)},\mathcal{Q}^{(n)}\}=0, \quad \{H^{(m)},H^{(n)}\}=0, \quad m,n = 0,1,\ldots,N-1.
\end{equation}
\end{enumerate}
\end{corollary}

\subsection{The Integro-Differential Limit}
\label{idl}
Let us now return to the full, integro-differential Bethe-Boltzmann equation
\begin{align}
\nonumber \partial_t\theta(x,t,k) + v[\theta](k)\partial_x\theta(x,t,k) &= 0 \\
v[\theta](k) &= \frac{\int dk' \, U(k,k')k'}{\int dk' \, U(k,k')}.
\label{BBC}
\end{align}
Viewing this as an infinite-dimensional system of hydrodynamic type in $(1+1)$D, a convenient choice of target space is $\mathcal{E}=L^2(\mathbb{R})$. By boundedness of the TBA kernel $\mathcal{K}(k,k')$, the dressing operator $U$, defined as in Eq. \eqref{dress}, exists for $\theta$ in some neighbourhood $\mathcal{A}$ of $0 \in \mathcal{E}$, upon which we assume that the system \eqref{BBC} is defined. We also introduce the kernel
\begin{equation}
\label{kern}
\alpha(k,k') = -\int dk''\,U(k,k'')\mathcal{K}(k'',k'),
\end{equation}
which generalizes the $\alpha_{ij}$ above. To us, the main difficulties in obtaining a rigorous demonstration of integrability for the system \eqref{BBC} are the following:
\begin{itemize}
\item The integrability conditions used to derive Tsar{\"e}v's key results \cite{Tsarev85} now hold at the level of functional derivatives, which makes their mathematical interpretation somewhat unclear.
\item The Riemannian geometry associated with brackets of hydrodynamic type is now infinite-dimensional.
\end{itemize}
Indeed, with a few notable exceptions \cite{KuperII,FerapontovMarshall}, the theory of infinite-dimensional systems of hydrodynamic type remains poorly understood \cite{Dorfman}. Nevertheless, upon replacing partial derivatives by functional derivatives, many properties of the integro-differential system \eqref{BBC} are demonstrable in exactly the same way as their finite-dimensional analogues. To us, this suggests that there exists an underlying integrable structure in the integro-differential limit. We begin by amassing some evidence for this claim, in the spirit of our earlier results.
\begin{prop}
\label{ctsWNS}
The characteristic velocities of the system \eqref{BBC} possess the \emph{continuum WNS property}:
\begin{enumerate}
\item \begin{align}\frac{\delta v(k)}{\delta \theta(k)} = 0, \quad k \in \mathbb{R}.\end{align}
\item \begin{align}\frac{\delta}{\delta\theta(k'')}\left[\frac{1}{v(k')-v(k)}\frac{\delta v(k)}{\delta \theta(k')}\right] = \frac{\delta}{\delta\theta(k')}\left[\frac{1}{v(k'')-v(k)}\frac{\delta v(k)}{\delta \theta(k'')}\right], \quad k\neq k'\neq k'' \neq k.\end{align}
\end{enumerate}
\end{prop}
From consideration of Prop. \ref{WNS} and Prop. \ref{wnscharges}, we deduce the following.
\begin{prop}
\begin{enumerate}
\item The system \eqref{BBC} has a countably infinite family of independent, commuting flows. These satisfy the continuum WNS property and are given by
\begin{equation}
w^{(n-1)}(k) = \frac{{Q^{(n)}}'(k)}{P'(k)},\quad n\in \mathbb{N}.
\end{equation}
\item There is a countably infinite family of first integrals of \eqref{BBC}, given by
\begin{equation}
\mathcal{Q}^{(0)}[\theta] = \int dx \int dk \, \rho(x,k), \quad \mathcal{Q}^{(n)}[\theta] = \int dx \int dk \, \rho(x,k)q^{(n)}(k), \quad n\in\mathbb{N},
\end{equation}
which correspond to the total particle number and the higher conserved charges of the Lieb-Liniger model\footnote{This, of course, is a \emph{sine qua non} \cite{BB1,BB2}, but the connection with classical integrability is new.}.
\end{enumerate}
\end{prop}
We can use these flows to solve for the hydrodynamic symmetries of the model; the resulting vector fields satisfy Tsar{\"e}v's condition \eqref{flow} expressed in terms of functional derivatives.
\begin{prop}
The system \eqref{BBC} has uncountably many hydrodynamic symmetries $w(k)$ solving the system
$$\frac{1}{w(k')-w(k)}\frac{\delta w(k)}{\delta \theta(k')} = \frac{\delta}{\delta \theta(k')} \log P'(k), \quad k \neq k',$$
parameterized by functions $f(\xi,k)$ of two variables, in terms of which
\begin{equation}
w(k) = \frac{1}{P'(k)} \int dk' \, U(k,k') \int dk'' \int^{\theta(k'')} d\xi \, \left[\delta(k'-k'')+\mathcal{K}(k',k'')\xi\right] \frac{1}{f(\xi,k'')}.
\label{inflows}
\end{equation}
\end{prop}
In turn, these symmetries give rise to an uncountable family of first integrals, as in Prop. \ref{intprop}, and one can check directly that these quantities are conserved. Based on these continuum analogues of our earlier results, we conjecture that the Bethe-Boltzmann equation is integrable in the same sense as its discretizations.
\begin{conj}
\label{conj1}
For two functionals $(I,J)$ of $\theta(x,k)$, define
\begin{equation}
\{I,J\} = \int dx \int dk \, dk' \frac{\delta I}{\delta \theta(x,k)} A(k,k') \frac{\delta J}{\delta \theta(x,k')},
\end{equation}
where the operator-valued kernel $A(k,k')$ is given by
\begin{align}
\label{ctsbracket}
\nonumber A(k,k') = &\frac{\delta(k-k')}{P'(k)^2}\left[\frac{d}{dx} - \int dk''\, \frac{\alpha(k,k'')P'(k'')}{P'(k)}\theta_x(x,k'')\right] + \frac{\alpha(k,k')}{P'(k)P'(k')}(\theta_x(x,k)-\theta_x(x,k'))\\
 - &\frac{\theta_x(x,k)}{P'(k)}\int dk''\,\alpha(k,k'')d_x^{-1}\alpha(k'',k')\frac{1}{P'(k')}\theta_x(x,k').
\end{align}
Then the Bethe-Boltzmann equation Eq. \eqref{BBC} is Hamiltonian with respect to the bracket of functionals \eqref{ctsbracket}.
\end{conj}
While we do not prove this here, owing to the technical difficulties mentioned above, this bracket seems to be the most natural continuum limit of the expression \eqref{discbrack} obtained above. For example, for a given functional $H[\theta]$, we have
\begin{align}
\nonumber \{\theta(x,k),H\} = &\frac{1}{P'(k)^2}\int dk' \left[\frac{\delta^2 H}{\delta \theta(k)\delta \theta(k')} -\frac{\alpha(k,k')P'(k')}{P'(k)} \frac{\delta H}{\delta \theta(k)} - \frac{\alpha(k,k')P'(k)}{P'(k')}\frac{\delta H}{\delta \theta(k')}\right] \theta_x(x,k') \\
&+ \frac{1}{P'(k)}\theta_x(x,k)\int dk' \left[\frac{\alpha(k,k')}{P(k')} \frac{\delta H}{\delta \theta(k')} - \int dk'' \ \alpha(k,k'') d_x^{-1} \alpha(k'',k')\theta_x(x,k') \frac{1}{P'(k')}\frac{\delta H}{\delta \theta(k')}\right],
\end{align}
and can show the following.
\begin{prop}
\label{sensible}
\begin{enumerate} 
\item The first integrals
\begin{equation}
\mathcal{Q}^{(0)}[\theta] = \int dx \int dk \, \rho(x,k), \quad \mathcal{Q}^{(n)}[\theta] = \int dx \int dk \, \rho(x,k)q^{(n)}(k), \quad n\in\mathbb{N},
\end{equation}
lie in the kernel of the bracket \eqref{ctsbracket}.
\item The first integrals $\{H^{(0)}, H^{(1)},\ldots\}$, given by
\begin{equation}
H^{(n)}[\theta] = - \pi \int dx \, \int dk \, \rho(x,k)\theta(x,k){q^{(n)}}'(k),
\end{equation}
generate the commuting flows $w^{(0)}, w^{(1)},\ldots$ at each point.
\item The Casimirs $\{\mathcal{Q}^{(n)}\}_{n=0}^{\infty}$ and Hamiltonians $\{H^{(n)}\}_{n=0}^{\infty}$ satisfy the algebra
\begin{equation}
\{\mathcal{Q}^{(m)},\mathcal{Q}^{(n)}\}=0, \quad \{H^{(m)},\mathcal{Q}^{(n)}\}=0, \quad \{H^{(m)},H^{(n)}\}=0, \quad m,n = 0,1,\ldots
\end{equation}
\end{enumerate}
\end{prop}
The proof of this result proceeds just as for Proposition \ref{theham}. It is striking that the algebra $[\mathbf{Q}^{(m)},\mathbf{Q}^{(n)}] = 0$ of conserved charges of the underlying quantum model appears to resurface twice over in the hydrodynamic brackets. We also note that whereas for $N$-point discretizations of the Bethe-Boltzmann theory, only $N$ of the Hamiltonians $H^{(n)}$ generate linearly independent flows, in the integro-differential limit there appear to be \emph{countably} many Hamiltonians in involution, which generate linearly independent flows on the target space. This relatively well-behaved subspace of the uncountably large space of hydrodynamic symmetries points towards an ``enriched'' notion of integrability in the integro-differential limit.

\section{Conclusion}

We have investigated in detail the question of integrability of the Bethe-Boltzmann equation, building upon the telegraphic observations of an earlier work \cite{BB6}. In particular, by proving rigorously various properties of a class of natural discretizations of this equation and extending some of these to the integro-differential limit, we have amassed a large body of evidence for its integrability. We have additionally demonstrated that the ``integrability class'' of the discretized Bethe-Boltzmann equation, as given by a reciprocal transformation, is precisely that of its non-interacting limit. This leads to a simple geometrical picture of quasiparticle dressing in the thermodynamic Bethe ansatz, which is possibly related to other recent developments in this area \cite{DS1}.

Several natural questions arise for further work. The most important of these is to prove that Eq. \eqref{ctsbracket} defines a Poisson bracket. To our knowledge, the closest approach to this bracket in the literature was obtained for the Vlasov equation in Ref. \cite{Vlasov}, and it seems likely that Eq. \eqref{ctsbracket} can be obtained by similar methods. However, to put this on a rigorous footing would require a ``hydrodynamic chain'' formalism for the Bethe-Boltzmann equation (c.f. \cite{Gibbons}), which is presently lacking. Alternatively, it might be possible to demonstrate integrability of the Bethe-Boltzmann equation indirectly, for example, by exhibiting a suitable infinite-dimensional reciprocal transformation from an associated linear system. This is because the results of Section \ref{idl} imply that the Bethe-Boltzmann equation lies in a class of infinite-dimensional linearly degenerate semi-Hamiltonian systems, generalizing those considered in Refs. \cite{Fera,FP}. As mentioned in the introduction, the kinetic equation for the classical hard-rod gas also lies in the class. It is plausible that such systems have appeared in wider physical settings related to kinetic theory, and we hope that the present work can motivate further research on this topic.

Another interesting question concerns the physical interpretation of the Hamiltonian and the hydrodynamic symmetries given above. To our knowledge, these quantities have not appeared in the literature on quantum integrable systems before and seem to imply an uncountable family of ``dynamical'' conservation laws, at least at the level of ballistic transport. It also seems significant that the Hamiltonians $H^{(n)}$, which generate the commuting flows $w^{(n)}(k) = {Q^{(n)}}'(k)/P'(k)$, lie in this non-trivial class of first integrals, unlike the conserved charges $\mathcal{Q}^{(n)}$ arising from quantum integrability. Meanwhile, the appearance of the $\mathcal{Q}^{(n)}$ as Casimirs of the hydrodynamic bracket suggests the intriguing possibility that this bracket could be used as a guide to incorporating new physical effects in the Bethe-Boltzmann equation.

Finally, it would be interesting to develop further the analogy between the integrable structures discovered in the El-Kamchatnov equation for the KdV soliton gas \cite{El3} and those obtained in the present work. For example, there might be non-trivial exact solutions to Eq. \eqref{crucialint} mirroring the exact solutions obtained in Ref. \cite{El3}. In general, we expect that the El-Kamchatnov equation should be integrable in the sense of the above results whenever the classical phase shifts $\Delta x(k,k')$ possess a well-defined thermodynamic Bethe ansatz, which seems to impose a non-trivial constraint on the underlying gas of classically integrable solitons.

Historically speaking, Poisson brackets of hydrodynamic type were discovered in the context of Whitham averaging of so-called ``soliton lattice'' solutions to classically integrable PDEs \cite{Whitham,Dub}. We find it rather satisfying that such Poisson brackets, rooted in the classical theory of integrability, re-emerge in the quasiparticle dynamics of quantum integrable systems.

\section{Acknowledgments}
We thank J. E. Moore for introducing us to the Bethe-Boltzmann equation and R. Vasseur and N. Y. Reshetikhin for encouraging us to pursue the study of its integrability. We are grateful to an anonymous referee for drawing Ref. \cite{Vlasov} to our attention. We would also like to thank the Institut d'{\'E}tudes Scientifique de Carg{\`e}se and the organizers of the workshop ``Exact Methods in Low-Dimensional Statistical Physics'' for their hospitality.
\bibliography{ClassIntBib}
\appendix


\clearpage

\section{Formulation for General Quantum Integrable Models}
In this Appendix, we formulate the results of the main text for arbitrary quantum integrable models with diagonal scattering.
\subsection{Hamiltonian Structure of Discretizations}
\label{GeneralDisc}
For a general integrable model, the main difference compared to the exposition above is that there are now generically $N_s \neq 1$ different species of quasiparticle to consider \cite{BB1,BB2,Takahashi}. In particular, this requires us to introduce a quasiparticle index. Thus, in the discretized case, the dressing equation now reads
\begin{equation}
\sum_{j,b} (\delta^{ab}_{ij} + T^{ab}_{ij}(\theta))Q'_{(b,j)}(\theta) = q'_{(a,i)},
\end{equation}
with $T^{ab}_{ij} = K^{ab}_{ij}(\theta)\theta^{(b,j)}$ for some kernel $K^{ab}_{ij}$ and
\begin{align}
\delta^{ab}_{ij} = \begin{cases} 1 & a=b,\ i=j \\ 0 & \rm{otherwise}\end{cases}.
\end{align}
The quasiparticle velocites now depend on the quasiparticle species, via
\begin{equation}
v_{(a,i)}(\theta) = E'_{(a,i)}(\theta)/P'_{(a,i)}(\theta), \quad a=1,2,\ldots,N_s,
\label{DBB3}
\end{equation}
and the discretized Bethe-Boltzmann equation becomes a system of $NN_s$ equations
\begin{equation}
\partial_t\theta^{(a,i)} + v_{(a,i)}(\theta)\partial_x\theta^{(a,i)} = 0, \quad i = 1,2,...,N,\ a=1,2,\ldots,N_s.
\label{DBB4}
\end{equation}
It is clear that structurally, these equations have precisely the form considered in Section \ref{sec3}, the only difference being that the index $i$ is replaced by pairs of indices $(a,i)$. For example, a natural choice of state-space metric is given by
\begin{equation}
g_{(a,i)(a,i)} = (P'_{(a,i)})^2,
\end{equation}
and the semi-Hamiltonian property arises from the identity
\begin{equation}
\Gamma^{(a,i)}_{(a,i)(b,j)} = \frac{\partial_{(b,j)}v_{(a,i)}}{v_{(b,j)}-v_{(a,i)}} = \partial_{(b,j)}\log{P'_{(a,i)}}.\end{equation}
However, there is now the difficulty that invertibility of the matrix of derivatives of bare charges, $q'$, must be considered on a case-by-case basis. A similar issue arises for the derivatives of dressed quasiparticle momenta, $P'_{(a,i)}$. While these are essentially artifacts of the discretization and make little difference in the continuum limit, they lead to a more restrictive version of Corollary \ref{maincorr} above.
\begin{corollary}
Suppose that the derivatives of dressed quasiparticle momenta with respect to rapidity, $P'_{(a,i)}(\lambda)$, are positive in some neighbourhood $\mathcal{S}$ of $\theta=0$. Then we have the following.
\begin{enumerate}
\item For any finite discretization length $N \geq 2$, the discretized Bethe-Boltzmann system, \eqref{DBB3} and \eqref{DBB4}, is Hamiltonian on $\mathcal{S}$, with respect to the non-local Poisson bracket
\begin{equation}
A^{(a,i)(b,j)} = g^{(a,i)(a,i)}\delta^{ab}_{ij}\frac{d}{dx} - g^{(a,i)(a,i)}\Gamma^{(b,j)}_{(a,i)(d,k)}\theta^{(d,k)}_x - \sum_{c=1}^{N_{s}}\sum_{m=1}^N\chi^{(c,m)}_{(a,i)}\theta^{(a,i)}_x d_x^{-1}\chi^{(c,m)}_{(b,j)}\theta^{(b,j)}_x,
\end{equation}
where the affinors $\{\chi^{(c,m)}_{(a,i)}\}$ are obtained by dressing the columns of the kernel, namely
\begin{equation}
\chi^{(c,m)}_{(a,i)}(\theta) = \frac{\sum_{b,k} U^{ab}_{ik}(\theta)K^{bc}_{km}}{P'_{(a,i)}(\theta)}.
\end{equation}
\item With respect to this bracket, the system \eqref{DBB4} can be written in Hamiltonian form as
\begin{equation}
\partial_t \theta^{(a,i)} = \{\theta^{(a,i)},H^{(1)}\},
\end{equation}
where
\begin{equation}
H^{(1)}[\theta] = -\frac{1}{2}\sum_{a,i} \int dx \, (\theta^{(a,i)})^2P'_{(a,i)}(\theta)e'_{(a,i)}
\end{equation}
\item The system \eqref{DBB4} has infinitely many hydrodynamic symmetries, given in terms of $NN_s$ functional degrees of freedom $f_{(c,k)}(\xi)$ by
\begin{equation}
w_{(a,i)}(\theta) = \frac{1}{P'_{(a,i)}(\theta)}\sum_{b,j,c,k}U^{ab}_{ij}(\theta)\int^{\theta^{(c,k)}} \frac{(\delta^{bc}_{jk} + K^{bc}_{jk}\xi)d\xi}{f_{(c,k)}(\xi)}.
\label{genflow2}
\end{equation}
Moreover, every smooth solution to this system is given locally by a solution to the generalized hodograph equations
\begin{equation}
w_{(a,i)}(\theta) = x-v_{(a,i)}(\theta)t,
\end{equation}
for some $w_{(a,i)}(\theta)$ of the form \eqref{genflow2}.
\end{enumerate}
\end{corollary}
\subsection{The Integro-Differential Limit}
\label{GeneralCts}
Allowing for $N_s$ quasiparticle species, and noting that the rapidities $\lambda$ no longer necessarily coincide with the quasimomentum $k$, we obtain the following system of $N_s$ equations:
\begin{align}
\nonumber \partial_t\theta^a(x,t,\lambda) + v_a[\theta](\lambda)\partial_x\theta^a(x,t,\lambda) &= 0 \\
v_a[\theta](\lambda) &= \frac{\sum_b \int d\lambda'\,U_{ab}(\lambda,\lambda')e'_b(\lambda')}{\sum_b \int d\lambda'\,U_{ab}(\lambda,\lambda')p'_b(\lambda')},
\label{BBC2}
\end{align}
where the integral kernel $U_{ab}$ satisfies
\begin{equation}
U_{ab}(\lambda,\lambda') + \sum_c \int d\lambda'' K_{ac}(\lambda,\lambda'')\theta^c(\lambda'')U_{cb}(\lambda'',\lambda') = \delta(\lambda-\lambda')\delta_{ab},
\end{equation}
and it is useful to define the kernel $\alpha$ by
\begin{equation}
\alpha_{ab}(\lambda,\lambda') = -\sum_c \int d\lambda''\, U_{ac}(\lambda,\lambda'')\mathcal{K}_{cb}(\lambda'',\lambda').
\end{equation}
It is trivial to extend the results above.
\begin{prop}
The characteristic velocities of the system \eqref{BBC2} possess the continuum analogue of the WNS property:
\begin{enumerate}
\item $$\frac{\delta v_a(\lambda)}{\delta \theta^a(\lambda)} = 0, \quad \lambda \in \mathbb{R},\ a=1,2,\ldots,N_s.$$
\item $$\frac{\delta}{\delta\theta^c(\lambda'')}\left[\frac{1}{v_b(\lambda')-v_a(\lambda)}\frac{\delta v_a(\lambda)}{\delta \theta^b(\lambda')}\right] = \frac{\delta}{\delta\theta^b(\lambda')}\left[\frac{1}{v_c(\lambda'')-v_a(\lambda)}\frac{\delta v_a(\lambda)}{\delta \theta^c(\lambda'')}\right]$$ for all $\lambda\neq \lambda'\neq \lambda'' \neq \lambda$ and $a,\ b,\ c$.
\end{enumerate}
\end{prop}
The charge densities $q^{(n)}(\lambda)$ now take on a quasiparticle index and give rise to dressed charge densities $Q_n(\lambda)$, whose derivatives with respect to $\lambda$ satisfy
\begin{equation}
{Q^{(n)}_a}'(\lambda) = \sum_b \int d\lambda' \, U_{ab}(\lambda,\lambda'){q^{(n)}_b}'(\lambda').
\end{equation}
\begin{prop}
\begin{enumerate}
\item The system \eqref{BBC2} has countably many WNS flows in involution, given by
\begin{equation}
w_a^{(n-1)}(\lambda) = \frac{{Q^{(n)}_a}'(\lambda)}{P_a'(\lambda)}, \quad n\in \mathbb{N}, \, a =1,2,\ldots,N_s.
\end{equation}
\item The system \eqref{BBC2} has a countable family of first integrals, with components
\begin{equation}
\mathcal{Q}^{(0)}_a = \int dx \int d\lambda \, \rho_a(x,\lambda), \quad \mathcal{Q}_a^{(n)} = \int dx \int d\lambda \, \rho_a(x,\lambda)q_a^{(n)}(\lambda), \quad n\in\mathbb{N}.
\end{equation}
\end{enumerate}
\end{prop}
We can again use the latter to solve for hydrodynamic symmetries of the model, obtaining the following:
\begin{prop}
The system \eqref{BBC2} has infinitely many hydrodynamic symmetries $w(k)$ solving the system
\begin{equation}
\frac{1}{w_b(\lambda')-w_a(\lambda)}\frac{\delta w_a(\lambda)}{\delta \theta^b(\lambda')} = \frac{\delta}{\delta \theta^b(\lambda')} \log P'_a(\lambda),
\end{equation}
parameterized by $N_s$ functions $f_a(\xi,\lambda')$ of two variables, in terms of which
\begin{equation}
w_a(\lambda) = \frac{1}{P_a'(\lambda)}\sum_{b,c} \int d\lambda' \, U_{ab}(\lambda,\lambda') \int d\lambda'' \, \int^{\theta^c(\lambda'')} d\xi \, [\delta(\lambda'-\lambda'')\delta_{bc}+\mathcal{K}_{bc}(\lambda',\lambda'')\xi] \frac{1}{f_c(\xi,\lambda'')}.
\label{inflows2}
\end{equation}
\end{prop}
The putative Poisson bracket now reads
\begin{equation}
\{I,J\} = \sum_{a,b=1}^{N_s} \int dx \int d\lambda \ d\lambda' \frac{\delta I}{\delta \theta^a(x,\lambda)} A^{ab}(\lambda,\lambda') \frac{\delta J}{\delta \theta^b(x,\lambda')},
\end{equation}
with kernel $A^{ab}(\lambda,\lambda')$ given by
\begin{align}
\nonumber A^{ab}(\lambda,\lambda') = &\frac{1}{P'_a(\lambda)^2}\left[\frac{d}{dx} - \sum_{c=1}^{N_s} \int d\lambda'\, \frac{\alpha_{ac}(\lambda,\lambda''){P'}_c(\lambda'')}{P'_a(\lambda)}\theta^c_x(x,\lambda'')\right] + \frac{\alpha_{ab}(\lambda,\lambda')}{P'_a(\lambda)P'_b(\lambda')}(\theta^a_x(x,\lambda)-\theta^b_x(x,\lambda'))\\
 - \sum_{c=1}^{N_s} &\frac{1}{P'_a(\lambda)}\theta^a_x(x,\lambda)\int d\lambda''\,\alpha_{ac}(\lambda,\lambda'')d_x^{-1}\alpha_{cb}(\lambda'',\lambda')\frac{1}{P_b(\lambda')}\theta^b_x(x,\lambda').
\end{align}
Analogues of Conjecture \ref{conj1} and Proposition \ref{sensible} may be formulated as above.
\end{document}